%% file: main.tex
\documentclass{article}
\usepackage{graphicx} 
\usepackage{import}
\usepackage{basic}
\usepackage{fullpage}

\title{Welfare Maximization in Bilateral Trade:\\ Improved Approximation Guarantees Beyond the Fixed Price Barrier \thanks{This was supported by Israel Science Foundation Grant 3129/25 and a BSF-NSF grant (BSF number: 2024785, NSF number: 2513504). The work of A. Shaulker was also supported by the Israel Science Foundation grant No. 301/24 and grant number 505/23.}}

\author{Shahar Dobzinski \thanks{Weizmann Institute of Science (shahar.dobzinski@weizmann.ac.il).} \and
 Ariel Shaulker \thanks{Hebrew University of Jerusalem (ariel.shaulker@huji.ac.il).}}

\begin{document}

\maketitle
\begin{abstract}
    \input{abstract}
\end{abstract}

\input{intro}

\input{preliminaries}

\input{BOreserve}

\input{seller-only}

\input{both-sides}

\bibliography{bilateral-trade}
\bibliographystyle{alpha}

\appendix
\input{lp-instance}

\end{document}

%% file: abstract.tex
We study the setting of welfare maximization in bilateral trade, where
the values of both the buyer and the seller are drawn from independent
distributions. Our goal is to maximize social welfare. In this
setting, fixed price mechanisms have been extensively studied. In a
fixed price mechanism, there is a price $p$ that depends only on the
distributions of the buyer and the seller. Trade occurs if and only if
the buyer’s value is at least $p$ and the seller’s value is at most
$p$. A long line of work has culminated in determining almost exactly
the approximation ratios achievable by fixed price mechanisms: there
exists a fixed price mechanism that obtains at least a $0.72$ fraction
of the social welfare, but no fixed price mechanism can guarantee more
than a $0.7381$ fraction of it \cite{CW-STOC, CSTOC}. No other incentive-compatible mechanism is known to beat the performance of fixed-price mechanisms in this setting.

This paper shows how to achieve a larger fraction of
the optimal welfare with other classes of mechanisms. Specifically, we study the buyer-offering
mechanism with a reserve price. In this mechanism, the buyer observes
its value and makes a take-it-or-leave-it offer to the seller, where
the offer is at least the reserve price. Beyond its simplicity, this
natural mechanism is attractive because the seller always has a
dominant strategy: accept the offer if its value is at most the offer,
and otherwise reject it. We show that there always exists a reserve price that guarantees a $\upperbound$ fraction of the social welfare. This not only improves upon the best previously known approximation guarantee for the problem, but also demonstrates that fixed-price mechanisms are not optimal in this setting.

%% file: intro.tex
\section{Introduction}

\subsubsection*{Background: Bilateral Trade and Fixed Price Mechanisms}

We study one of the simplest and most basic settings in mechanism design: bilateral trade, where the values of the buyer and the seller are drawn from independent distributions. The setting consists of a single item, a seller, and a buyer, where the seller initially holds the item. The value of the seller for the item, denoted $v_s$, is private and drawn from a publicly known distribution $\mathcal F_s$. Similarly, the private value of the buyer for the item is denoted $v_b$ and is drawn from a publicly known distribution $\mathcal F_b$.

Our goal in this paper is to facilitate trade between the buyer and the seller in order to maximize social welfare, which is simply $\max(v_s, v_b)$. If trade occurs, the seller receives a payment $p_s$ and the buyer pays $p_b$. Each participant seeks to maximize their profit: if trade occurs, the seller’s profit is $p_s$ and the buyer’s profit is $v_b - p_b$. If no trade occurs, then $p_b = p_s = 0$; in this case, the seller’s profit is $v_s$ and the buyer’s profit is $0$. We focus on incentive-compatible mechanisms that are ex post individually rational (if trade occurs, $p_s \ge v_s$ and $v_b \ge p_b$) and budget balanced ($p_b = p_s$).

The bilateral trade problem was studied in the seminal work of Myerson and Satterthwaite \cite{Myerson-Satterthwaite}. They showed that unless the distributions $\mathcal F_s$ and $\mathcal F_b$ are very degenerate, as characterized by a precise technical condition, the optimal ``first-best'' solution is not implementable, even by mechanisms that are incentive compatible in the Bayes--Nash sense.

Since the optimal solution is not implementable, work in algorithmic mechanism design on this problem has focused on approximation guarantees. In particular, fixed-price mechanisms have been extensively studied. In a fixed-price mechanism, there is a price $p$ that depends only on the distributions $\mathcal F_s$ and $\mathcal F_b$, but not on the realized values $v_s$ and $v_b$. In such a mechanism, each participant has a dominant strategy: the seller accepts if and only if $v_s \le p$, and the buyer accepts if and only if $v_b \ge p$.

In fact, a well-known observation is that every dominant-strategy mechanism for bilateral trade must be a fixed-price mechanism. To see this, fix a value $v_b$ of the buyer. Incentive compatibility implies that the payment $p_s$ received by the seller cannot depend on its reported value $v_s$; otherwise, the seller could profitably misreport its value to obtain a higher payment whenever trade occurs. A symmetric argument shows that the payment $p_b$ made by the buyer cannot depend on $v_b$. Since the mechanism is budget balanced and satisfies $p_s = p_b$, the trade price cannot depend on either agent’s realized value, and therefore must be fixed.

Fixed-price mechanisms have been shown to be quite powerful in this setting. Blumrosen and Dobzinski~\cite{BD14} showed that fixed-price mechanisms guarantee at least one half of the optimal social welfare for every choice of distributions. They also showed that a slightly larger fraction $\frac {28} {55}$ can be guaranteed by a more careful analysis. This was subsequently improved to $\frac {13} {25}$~\cite{italians}, then to $1-\frac{1}{e}$~\cite{BD21}. Even this bound was further refined to $1-\frac 1 e+0.0001$~\cite{KPV}, and finally tightened to an almost optimal factor of approximately $0.72$~\cite{CW-STOC,CSTOC}. These two last papers also show that no fixed-price mechanism can obtain a fraction larger than $0.7381$. 

A substantial body of work has also focused on approximating the gains from trade, which measures the improvement by trade: $\max (v_b-v_s,0)$. McAfee~\cite{mcaffee} showed that for certain distributions of buyer and seller values, a fixed-price mechanism can achieve at least one half of the optimal gains from trade. Blumrosen and Dobzinski~\cite{BD14} proved that for any constant fraction, there exist distributions for which no fixed-price mechanism can guarantee that fraction of the optimal gains from trade.

To bypass this limitation, Blumrosen and Mizrahi~\cite{BM} introduced the seller-offering mechanism, in which the seller makes a profit-maximizing take-it-or-leave-it offer to the buyer based on its value~$v_s$. In this mechanism, only the buyer has a dominant strategy. Nevertheless, they showed that when the buyer's value distribution satisfies the monotone hazard rate condition, the seller-offering mechanism is Bayesian incentive compatible and guarantees a fraction of $\frac{1}{e}$ of the optimal gains from trade. Subsequently, Brustle et al.~\cite{Bru} showed that selecting the better outcome between the seller-offering mechanism and the buyer-offering mechanism, where the buyer makes a profit-maximizing take-it-or-leave-it offer to the seller, recovers at least half of the gains from trade achievable by any incentive-compatible mechanism. In a major advance, Deng et al.~\cite{GDTCA} proved that the better of the two offering mechanisms guarantees about $0.122$ of the optimal gains from trade. This constant was later improved to $0.317$ by Fei~\cite{IGFT}.

Welfare guarantees in bilateral trade were also studied in the context of correlated values \cite{DS24} and interdependent values \cite{DEGST25}. 


\subsubsection*{Our Contribution: Buyer-Offering Mechanisms and One-Sided Dominant Strategy Mechanisms}

Our goal in this paper is to design incentive-compatible mechanisms that improve upon the welfare guarantees achievable by fixed-price mechanisms in the basic setting where valuations are drawn from independent distributions. Moreover, prior to this
work, fixed-price mechanisms provided the best known welfare guarantees for this
setting, and therefore our improvement over fixed-price mechanisms also improves the
best known welfare guarantees for bilateral trade with independent values.

Recall that fixed-price mechanisms coincide exactly with the class of dominant-strategy mechanisms; therefore, obtaining stronger guarantees requires considering weaker solution concepts. A natural candidate is the class of Bayes–Nash incentive-compatible mechanisms. However, mechanisms in this class are often unnatural, and their equilibria tend to depend critically on both participants having precise knowledge of the underlying distributions. If even one participant has slightly incorrect information about one of the distributions, no meaningful guarantees may remain, and it is possible for the performance of the mechanism to deteriorate significantly.

Instead, we restrict our attention to Bayes-Nash incentive-compatible mechanisms in which one participant has a dominant strategy: one-sided dominant-strategy mechanisms. Our first observation is a characterization:


\headline{Observation: }{\begin{itemize}
\item Let $\mathcal M$ be a one-sided mechanism for bilateral trade where the seller has a dominant strategy. Then, there is a set of offers $\buyeroffers$ such that $\mathcal M$ can be implemented as a buyer-offering mechanism where for every $v_b$ the buyer selects its profit-maximizing offer from the set $\buyeroffers$.
\item Similarly, let $\mathcal M$ be a one-sided mechanism for bilateral trade where the buyer has a dominant strategy. Then, there is a set of offers $\selleroffers$ such that $\mathcal M$ can be implemented as a seller-offering mechanism where for every $v_s$ the seller selects its profit-maximizing offer from the set $\selleroffers$.
\end{itemize}
}

Note the strategic attractiveness of these mechanisms to the players: one side simply makes the best possible offer according to the information that it possesses, and the other side just accepts or rejects the offer. For example, in buyer-offering mechanisms this implies, in particular, that the players do not have to agree on the details of the distributions of the seller or the buyer, and neither of the players has to even hold any knowledge of the distribution of the buyer. It is only used to determine the set $\buyeroffers$ of possible offers. 

The buyer offering mechanism -- with no restrictions on the set of possible offers --  was also shown to be useful in the context of welfare maximization: its expected social welfare is at least $1 - \frac 1 e$ of the optimal social welfare, even if the values $v_s,v_b$ are drawn from a joint distribution rather than independently\footnote{Interestingly, the seller offering does not guarantee in general any constant fraction of the optimal social welfare: consider a seller with value $0$ and a buyer whose value is distributed according to an equal revenue distribution.} \cite{DS24}. Unfortunately, it can be shown that the buyer offering mechanism cannot guarantee a fraction larger than $1-\frac 1 e$ of the optimal welfare even if the values are independently drawn.

We observe that the buyer offering mechanism fails to provide a good fraction of the optimal welfare when the offers the buyer makes are too low. That is, in order to maximize its profit, the buyer prefers to decrease the probability of trade in exchange for a higher profit when a trade is made. A natural solution to this would be, of course, to restrict the set of possible offers by setting a reserve price so that the buyer cannot make low offers. 

Buyer offering mechanisms with reserves can also be observed in practice. In many markets, buyers are allowed to make take-it-or-leave-it
offers, but these offers are subject to a minimum acceptable price. For example,
Priceline's ``Name Your Own Price'' model allows buyers to submit bids for
services such as hotels, flights, and car rentals, where bids below an internal
reserve are rejected. More generally, reserve prices are a standard way to
prevent inefficiently low offers while preserving the simple interaction pattern
in which one side proposes a price and the other side only accepts or rejects.

\subsubsection*{Our Contribution: Finding the Right Reserve(s)}

Setting an appropriate reserve turns out to be a challenging task. Since the reserve cannot depend on the seller's value, if we set it too low, it becomes ineffective when the buyer's value is high. In contrast, if the reserve is high, a buyer with a value below it will not be able to trade at all. The key obstacle is therefore to understand whether there is a reserve price that balances the two competing effects.

The main technical tool that will aid us in proving the existence of a good reserve and then finding it is the following lemma:

\headline{The Tradeoff Lemma: } {Consider a buyer offering mechanism with some reserve price $r$. Let $\alpha_L$ denote the limit of the optimal welfare extracted (conditioned on the value of the buyer being $v_b$) when the value of the buyer $v_b$ approaches $r$ from below. Then, when $v_b>r$, the fraction of the optimal welfare extracted (conditioned on the value of the buyer being $v_b$) is $f(\alpha_L)$.
}

We defer the specification of the function $f(\cdot)$ to the technical parts of the paper. Before analyzing the implications of the tradeoff lemma, let us explain its statement. Fix a reserve price $r$. Conditioned on fixing the value of the buyer $v_b<r$, the value of the optimal solution is $ v_b\cdot \Pr[v_s\leq v_b]+ E[v_s|v_s>v_b]\cdot \Pr[v_s>v_b]$. Since $v_b<r$ the mechanism never trades. The expected welfare of the mechanism is simply $E[v_s]$. The ratio $\frac {E[v_s]} { v_b\cdot \Pr[v_s\leq v_b]+ E[v_s|v_s>v_b]\cdot \Pr[v_s>v_b]}$ is minimized when $v_b$ approaches $r$ from below, so $\alpha_L=\frac {E[v_s]} { r\cdot \Pr[v_s\leq r]+ E[v_s|v_s>r]\cdot \Pr[v_s>r]}$. The tradeoff lemma now tells us how low the fraction of the welfare guaranteed can be for each $v_b>r$, i.e., when the value of the buyer is above the reserve. This guarantee holds for every $v_b$, so, importantly, it is independent of the distribution of the values of the buyer.

A first application of the lemma is the following:

\headline{Theorem: } {For every distribution $\mathcal F_s$ of the seller, there is a reserve $r$, such that for every $v_b$ the fraction of the welfare that the buyer-offering mechanism with reserve $r$ extracts (conditioned on the value of the buyer being $v_b$) is at least $0.71$, regardless of the distribution $\mathcal F_b$.}

Note the appeal of this theorem: even if the buyer's distribution is not known at all, the theorem guarantees that we can use only the distribution of the seller to set a reserve price. Moreover, the approximation guarantee is stronger than usual: it holds for every $v_b$ (and not in expectation over the $\mathcal F_b$). We also prove that this ratio is very close to what is possible to achieve with this class of mechanisms:

\headline{Theorem: } {There is a distribution of the seller $\mathcal F_s$, such that for every reserve price $r$, there is a value $v_b$ of the buyer such that the fraction of the optimal welfare that the buyer offering mechanism with reserve $r$ guarantees (conditioned on the value of the buyer being $v_b$) is at most $0.7159$.}

We also point out that deterministic fixed-price mechanisms that are not aware of the distribution of the buyer cannot recover more than $\frac 1 2$  of the optimal social welfare \cite{BD14}, so our (deterministic) mechanism already demonstrates how to beat fixed-price mechanisms. Yet, the performance of this mechanism is still slightly below $0.72$, which is the fraction of the optimal welfare that the best currently fixed-price mechanism guarantees. We thus also use the distribution of the buyer in order to beat this ratio:

\headline{The Main Result: } {For every two distributions $\mathcal F_s, \mathcal F_b$ of the seller and the buyer, there is a reserve $r$, such that the expected welfare of the buyer-offering mechanism with reserve $r$ is at least a fraction of $\upperbound$ of the welfare of the optimal social welfare.}

This ratio not only beats the best currently known fixed-price mechanism, but it is also better than the $0.7381$ bound that is the provable limit of what can be achieved by any fixed-price mechanism \cite{CW-STOC, CSTOC}. Hence, this mechanism is the first to beat fixed-price mechanisms in bilateral trade. Since the mechanisms of \cite{CW-STOC,CSTOC} are not only the best-known fixed-price mechanisms but also the best-known incentive-compatible mechanisms in general, the mechanism we present achieves the best guarantee currently known among all incentive-compatible mechanisms.

The proof of this last result is quite delicate. We use the tradeoff lemma to find many points $(\alpha_L,f(\alpha_L))$. Then, we further observe that for $v_b>r$ the guarantee can get better than what can be guaranteed by the tradeoff lemma as a function of the distance of $v_b$ from $r$: when the distance $|v_b-r|$ is small, the mechanism either trades the item or the value of the seller is close to or bigger than the value of the buyer, so trade is not necessary for good approximation. We use a variant of the tradeoff lemma to quantify this intuition. Together, this gives us, for every reserve $r$, a lower bound on the approximation ratio for every value $v_b$. The final step is to write a factor-revealing linear program that takes into account the distribution of the buyer $\mathcal F_b$. By considering specific numerical values, the LP gives us that there is always a reserve price that extracts at least a fraction of $\upperbound$ of the optimal social welfare, as needed. We note that quadratic programs were already used by~\cite{CW-STOC, CSTOC} to determine the approximation ratio of fixed-price mechanisms. However, despite considering a more involved class of mechanisms, we are able to formulate our problem as a linear program, which allows us to guarantee improved performance by increasing the size of the program while remaining within the limits of tractability.

\subsubsection*{Conclusion and Open Questions}

In this paper, we imporve over the best known approximation guarantee known for bilateral trade when the values of the participants are drawn from independent distributions. We achieve this by overcoming the limitations of dominant-strategy mechanisms by considering one-sided dominant-strategy mechanisms, thereby sacrificing as little as possible from a strategic standpoint. Several open questions remain. The most immediate one is to perform the analysis of the factor-revealing linear program in a more analytic way that will improve the approximation ratio.

It is also natural to ask for the exact performance guarantees achievable by one-sided dominant-strategy mechanisms and by general Bayes–Nash incentive-compatible mechanisms. We note that it is known that no deterministic Bayes–Nash mechanism can achieve an approximation ratio better than $0.89$~\cite{DS24}. 

More broadly, our results highlight the importance of identifying natural subclasses of Bayes–Nash incentive-compatible mechanisms that admit strong performance guarantees. Finding additional such classes, both within and beyond the realm of bilateral trade, is an important direction for future work.

%% file: preliminaries.tex
\section{Preliminaries}

In the bilateral trade problem, there are two players: a seller and a buyer. The seller owns an indivisible item and has value $v_s$ for it, while the buyer’s value for the item is $v_b$. The values $v_b$ and $v_s$ are drawn independently from distributions $\distbuyer$ and $\distribution$, respectively. That is, $v_b \sim \distbuyer$ and $v_s \sim \distribution$.  
We denote by $\suppbuyer$ and $\suppseller$ the supports of the buyer’s and seller’s distributions, respectively. As is common in the literature, we assume that $\suppbuyer , \suppseller \subseteq [0,1]$.

Let $\mechanism = (x,p)$ be a (deterministic) direct mechanism for the bilateral trade problem. The function  
$x : \suppseller \times \suppbuyer \to \{0,1\}$ is the allocation rule, which decides for each pair of values whether trade occurs.
The function $p : \suppseller \times \suppbuyer \to \mathbb{R}$ is the payment rule, which assigns to each pair of values the monetary transfer from the buyer to the seller.
We consider incentive-compatible and individually rational mechanisms, as defined next.

\begin{definition}
A mechanism $\mechanism= (\allocation, \payment)$ is dominant-strategy incentive-compatible for the seller if for every $v_s, v_s' \in \suppseller$ and every $v_b \in \suppbuyer$, it holds that
\[
 p(v_s, v_b)+ v_s \cdot (1-x(v_s, v_b))
\;\ge\;
 p(v_s', v_b)+ v_s \cdot (1-x(v_s', v_b) )
\]
\end{definition}

\begin{definition}
A mechanism $\mechanism= (\allocation, \payment)$ is \emph{dominant-strategy incentive-compatible} for the buyer if for every $v_b, v_b' \in \suppbuyer$ and every $v_s \in \suppseller$, it holds that
\[
 v_b \cdot x(v_s, v_b) - p(v_s, v_b)
\;\ge\;
v_b \cdot x(v_s, v_b') - p(v_s, v_b').
\]
\end{definition}

\begin{definition}
A mechanism $\mechanism= (\allocation, \payment)$ is \emph{Bayesian incentive-compatible} for the seller if for every $v_s, v_s' \in \suppseller$, it holds that
\[
\mathbb{E}_{v_b \sim \distbuyer} [  p(v_s, v_b)+(1-x(v_s, v_b))\cdot v_s  ]
\;\ge\;
\mathbb{E}_{v_b \sim \distbuyer} [  p(v_s', v_b)+(1-  x(v_s', v_b))\cdot v_s ]
\]
\end{definition}

\begin{definition}
A mechanism $\mechanism= (\allocation, \payment)$ is \emph{Bayesian incentive-compatible} for the buyer if for every $v_b, v_b' \in \suppbuyer$, it holds that
\[
\mathbb{E}_{v_s \sim \distribution} [v_b \cdot x(v_s, v_b) - p(v_s, v_b)]
\;\ge\;
\mathbb{E}_{v_s \sim \distribution} [v_b \cdot x(v_s, v_b') - p(v_s, v_b')]
\]
\end{definition}

\begin{definition}
A mechanism $\mechanism= (\allocation, \payment)$ is ex-post individually rational if for every $v_b \in \suppbuyer$ and every $v_s \in \suppseller$, it holds that
\[
v_b \cdot x(v_s, v_b)  \geq p(v_s, v_b) \quad \text{ and } \quad
  p(v_s, v_b) \geq v_s \cdot (1- x(v_s, v_b))
\]
\end{definition}

If there is no trade we assume that $p(v_s, v_b)=0$.

We say that a mechanism is \emph{dominant-strategy incentive-compatible} if it is dominant-strategy incentive-compatible for both players, and \emph{Bayesian incentive-compatible (BIC)} if it is Bayesian incentive-compatible for both players. 
We say that a mechanism is \emph{one-sided DSIC for the seller} if it is dominant-strategy incentive-compatible for the seller and Bayesian incentive-compatible for the buyer, and \emph{one-sided DSIC for the buyer} if it is dominant-strategy incentive-compatible for the buyer and Bayesian incentive-compatible for the seller.
More generally, we say that a mechanism is \emph{one-sided} if it is dominant-strategy incentive-compatible for one player and Bayesian incentive-compatible for the other.


The optimal welfare of an instance of the bilateral trade problem is
\[
\opt{\distribution, \distbuyer}
\;=\;
\mathbb{E}_{v_s \sim \distribution,\; v_b \sim \distbuyer}
\big[ \max \{ v_s, v_b \} \big]
\]
The welfare of a mechanism $\mechanism = (\allocation, \payment)$ for an instance of the bilateral trade problem is denoted
\[
\mechanism_{\distribution, \distbuyer}
\;=\;
\mathbb{E}_{v_s \sim \distribution,\; v_b \sim \distbuyer}
\big[ \allocation(v_s, v_b)\cdot v_b
+ \big(1-\allocation(v_s, v_b)\big)\cdot v_s \big]
\]

Similarly, for a fixed buyer value $v_b$, we denote by $\opt{\distribution}(v_b) = \mathbb{E}_{v_s \sim \distribution}[\max\{v_s, v_b\}]$ and $\mechanism_{\distribution}(v_b) = \mathbb{E}_{v_s \sim \distribution}[x(v_s, v_b)\cdot v_b + (1-x(v_s, v_b))\cdot v_s]$ the expected welfare of the optimal allocation and of the mechanism $\mechanism$, respectively, when the buyer's value is $v_b$ and the seller's value is drawn from $\distribution$.

In this paper, we consider the buyer-offering mechanism with a reserve $\reserve$, denoted by $\bomech{\reserve}$. In the buyer-offering mechanism with reserve $\reserve$, the buyer with value $v_b$ makes a take-it-or-leave-it-offer of price $\price \geq \reserve$ to the seller, where the price $\price$ maximizes the buyer's profit among all prices of at least $\reserve$.


\subsection{Characterizing One-sided Dominant Strategy Mechanisms}

The next proposition shows that every mechanism in which one side has a dominant strategy can be implemented as a buyer-offering mechanism.

\begin{proposition}\label{proposition:one-sided-char}
Let $M$ be a deterministic Bayesian incentive-compatible bilateral-trade mechanism in which the seller has a dominant strategy. Then there exists a set of take-it-or-leave-it offers $\buyeroffers$ (each offer is a price in $[0,1]$) such that $M$ can be implemented as a buyer-offering mechanism. Symmetrically, any mechanism in which the buyer has a dominant strategy can be implemented as a seller-offering mechanism.
\end{proposition}

\begin{proof}
Fix a mechanism $M$ in which the seller has a dominant strategy.

Fix a buyer with value $v_b$. Because the seller has a dominant strategy, the payment received upon trade cannot depend on the seller’s report. Otherwise, the seller could profit from adjusting the payment while still trading the item. Thus, when the value of the buyer is $v_b$, whenever trade occurs, the seller receives a fixed payment $p(v_b)$. I.e., fixing $v_b$, $M$ is equivalent to posting a take-it-or-leave-it offer at price $p(v_b)$ to the seller.

Let $\buyeroffers= \{\, p( v_b) :  v_b\in[0,1] \,\}$. The original mechanism $M$ is implemented by allowing the buyer, upon observing its value $v_b$, to choose the report $v_b$ that maximizes its profit. Equivalently, the buyer selects the offer $p\in \buyeroffers$ that maximizes its profit. Given any such offer, the seller’s dominant strategy is to accept if and only if $v_s\le p$. Hence $M$ is implementable as a buyer-offering mechanism.

The argument with roles reversed yields the symmetric statement for mechanisms in which the buyer has a dominant strategy.
\end{proof}

%% file: BOreserve.tex
    





\section{Tools for Analyzing Buyer-Offering Mechanisms}

This section introduces two technical tools that will be applied in Sections \ref{sec-seller-only} and \ref{sec-main-result} to prove the existence of good buyer-offering mechanisms. The first and main tool is the tradeoff lemma (Lemma \ref{lemma:assymetric}), which allows us to bound the approximation ratio of values above the reserve when given a bound on the approximation ratio of values below the reserve. The second lemma is the Neighborhood Lemma (Lemma \ref{lemma:distance}), which bounds the distance between two reserves, if we know the approximation ratio when the value of the buyer is below these reserves.

\subsection{The Tradeoff Lemma}

\begin{lemma}[The Tradeoff Lemma]\label{lemma:assymetric}
Let $\distribution$ be a distribution over the value of the seller and let $\appbelow \in (0,1)$.
Denote by $\reserve$ a value of the buyer such that the approximation ratio when never trading the item conditioned on the buyer's value being $r$ is at most $\appbelow$. Fix $v_b>\reserve$. Let $\lambda_{v_b}$ be such that
$\frac{\reserve}{v_b} \geq  \lambda_{v_b}$. Then, the approximation ratio of the buyer-offering with reserve $r$ (conditioned on the buyer's value being $v_b$) is at least as large as the solution to the optimization problem:
\begin{equation}\nonumber
 \min_{1-\appbelow \leq x \leq 1, \lambda_{v_b} \leq \ell \leq 1}   1+x(1-\ell)\ln{x} \text{ s.t. } \ell \geq \frac{1-x+x\ln{x}}{\frac{x\cdot\appbelow}{1-\appbelow}+x\ln{x}}\\
\end{equation}
Moreover, the solution to this optimization problem is obtained at the boundaries of the $\ell$ constraints, i.e., it is obtained when $\ell =\lambda_{v_b}$ or when $\ell = \frac{1-x+x\ln{x}}{\frac{x\cdot\appbelow}{1-\appbelow}+x\ln{x}}$.
\end{lemma}

The optimization problem in the tradeoff lemma might be hard to work with. The two corollaries below provide simpler optimization problems, at the cost of a looser bound. The two corollaries differ in which constraint they relax. The first (Corollary~\ref{cor:threshold-lemma-no-bound-on-v-b}) is useful when we have a poor bound, or no bound at all, on the value 
$v_b$ relative to the reserve $\reserve$. The second (Corollary~\ref{cor:threshold-lemma-bound-on-v-b}) is useful when we have a good bound on the value $v_b$ relative to the reserve $\reserve$. After introducing the corollaries, we prove the tradeoff lemma.

\begin{corollary}\label{cor:threshold-lemma-no-bound-on-v-b}
    Let $\distribution, \appbelow, \reserve, v_b$ be as in the tradeoff lemma. The approximation ratio of the buyer-offering with reserve $r$ (conditioned on the buyer's value being $v_b$) is at least as large as the solution to the optimization problem:
    $$
    \min_{1-\appbelow \leq x\leq 1} 1+x\ln x- \frac{\ln x \left(1-x+x\ln x\right)}{\ln x +\frac{\appbelow}{1-\appbelow}}
    $$
\end{corollary}

\begin{proof}[Proof of Corollary~\ref{cor:threshold-lemma-no-bound-on-v-b}]
 Let $\lambda_{v_b} = 0$. By the lemma, the approximation ratio of the buyer-offering mechanism with reserve $r$ (when the buyer's value is $v_b$) is at least the solution to the optimization problem:
 \begin{equation}\nonumber
 \min_{1-\appbelow \leq x \leq 1, 0 \leq \ell \leq 1}   1+x(1-\ell)\ln{x} \text{ s.t. } \ell \geq \frac{1-x+x\ln{x}}{\frac{x\cdot\appbelow}{1-\appbelow}+x\ln{x}}\\
\end{equation}
The lemma also states that the solution to this optimization problem is attained at the boundaries of the $\ell$ constraints. Simple algebra shows that $0 \leq \frac{1-x+x\ln{x}}{\frac{x\cdot\appbelow}{1-\appbelow}+x\ln{x}}$, since $x\geq 1-\appbelow$. Note that we get equality only when $x=1$. In this case, the objective value does not depend on $l$ at all. Therefore, we may assume that the solution is obtained when $\ell = \frac{1-x+x\ln{x}}{\frac{x\cdot\appbelow}{1-\appbelow}+x\ln{x}}$, which results in the following optimization problem:
$$
\min_{1-\appbelow \leq x\leq 1} 1+x\ln x- \frac{\ln x \left(1-x+x\ln x\right)}{\ln x +\frac{\appbelow}{1-\appbelow}}
$$
 \end{proof}

\begin{corollary}\label{cor:threshold-lemma-bound-on-v-b}
    Let $\distribution, \appbelow, \reserve, v_b$, and $\lambda_{v_b }> 0$ be as in the tradeoff lemma. The approximation ratio of the buyer-offering with reserve $r$ (conditioned on the buyer's value being $v_b$) is at least:
    $$
    \appabove(\appbelow, \lambda_{v_b}) = \begin{cases}
        1-\frac{(1-\lambda_{v_b})}{e} & \text{ if } 1-\appbelow \leq \frac{1}{e} \\
        1+(1-\appbelow)\cdot(1-\lambda_{v_b})\cdot\lnn{1-\appbelow} & \text{ if } 1-\appbelow > \frac{1}{e}
    \end{cases}
    $$
    
\end{corollary}


\begin{proof}[Proof of Corollary~\ref{cor:threshold-lemma-bound-on-v-b}]
    By the tradeoff lemma, the approximation ratio of the buyer-offering mechanism with reserve $r$ (when the buyer's value is $v_b$) is at least the solution to the optimization problem:
 \begin{equation}\nonumber
 \min_{1-\appbelow \leq x \leq 1, \lambda_{v_b} \leq \ell \leq 1}   1+x(1-\ell)\ln{x} \text{ s.t. } \ell \geq \frac{1-x+x\ln{x}}{\frac{x\cdot\appbelow}{1-\appbelow}+x\ln{x}}\\
\end{equation}
We consider two cases. In the first case, $\lambda \geq \frac{1-x+x\ln{x}}{\frac{x\cdot\appbelow}{1-\appbelow}+x\ln{x}}$. In this case, the constraints on $\ell$ reduce to $\ell \geq \lambda_{v_b}$. Since the objective function is monotone non-decreasing in $\ell$, the minimum is attained when $\ell$ is as small as possible, that is, 
\begin{equation}\label{problem:only-distance}
\min_{1-\appbelow \leq x \leq 1}   1+x(1-\lambda_{v_b})\ln{x}
\end{equation}
In the second case, $\lambda < \frac{1-x+x\ln{x}}{\frac{x\cdot\appbelow}{1-\appbelow}+x\ln{x}}$. Thus, retaining only the constraint $\ell \geq \lambda_{v_b}$ relaxes the problem and can only decrease the optimal value. As in the first case, this leads to the optimization problem in~\eqref{problem:only-distance}.
With the problem reduced to a single variable, we now solve it analytically.
Let $g(x) = 1+x(1-\lambda_{v_b})\ln{x}$. Its derivative is $g'(x)=(1-\lambda_{v_b})\cdot(\ln{x} +1)$. Since $\lambda_{v_b} < 1$ (as $v_b > \reserve$), the derivative equals $0$ only when $x = \frac{1}{e}$, and this value corresponds to a minimum of $g(x)$. Hence, if $1-\appbelow \leq \frac{1}{e}$, the solution to the optimization problem is attained at $x=\frac{1}{e}$, with value $1-\frac{(1-\lambda_{v_b})}{e}$. When $1-\appbelow > \frac{1}{e}$, the minimum is attained at $x=1-\appbelow$, with value $1+(1-\appbelow)\cdot(1-\lambda_{v_b})\cdot\lnn{1-\appbelow}$.
\end{proof}

\begin{proof}[Proof of Lemma~\ref{lemma:assymetric}]
We start with the following claim:


\begin{claim}\label{claim:price-equal-reserve}
There exists a distribution $\worsedis$ over the value of the seller with the following properties:
\begin{enumerate}
    \item\label{property-p=r} 
    When the distribution is $\worsedis$, $r$ has the highest profit for a buyer with value $v_b$ among all offers that are at least $\reserve$.
    \item When the distribution is $\worsedis$, the approximation ratio of the buyer-offering with reserve $\reserve$ when the buyer's value is $v_b$ is no better than the approximation ratio when the distribution is $\distribution$ at $v_b$.
    \item When the distribution is $\worsedis$ and the value of the buyer is $\reserve$, the approximation ratio of the mechanism that never trades is at most $\appbelow$.
    \item $\worsedis(v_b)=1$.
\end{enumerate}

\end{claim}
\begin{proof}[Proof of Claim \ref{claim:price-equal-reserve}]
Let $\price\geq \reserve$ be the optimal offer of a buyer with value $v_b$ in the buyer-offering with reserve $\reserve$ mechanism when the seller's distribution is $\distribution$. If $\price = \reserve$ then property \ref{property-p=r} holds. Assume that $\price > \reserve$. Since $\price > \reserve$, $\distribution(\price) > \distribution(\reserve)$. We construct a distribution $\constructiondis$ that maintains the bounds on the approximation ratio and satisfies property \ref{property-p=r}. $\constructiondis$ is the same as $\distribution$ except on the interval $[\reserve,\price]$, where we shift all the seller's mass to the point $r$, so $\constructiondis(\reserve) = \constructiondis(\price)$. 

Note that when the buyer's value is $v_b$, the optimal offer when the distribution of the seller is $\constructiondis$ is $\reserve$: shifting the mass in $[\reserve,\price]$ to $\reserve$ does not change the profits of any offer that is at least $\price$, so $\price$ still dominates them, while now it is clear the $\reserve$ is more profitable than $\price$ as the probability of sale is identical but the amount the buyer gets in every sale is strictly larger for $\reserve$. Furthermore, the approximation ratio at $v_b$ did not improve, as the welfare of the buyer-offering mechanism with reserve $\reserve$ remains the same and the optimal welfare has not increased. 
Similarly, for when the value of the buyer is $\reserve$, the approximation ratio of the mechanism that never trades the item could only get worse, since there is no trade and we have not increased $E[v_s]$. 

For the last property $\worsedis(v_b) =1$. If $\distribution(v_b) = \constructiondis(v_b) < 1$, we shift the seller's mass above $v_b$ to $v_s = \reserve$. This does not improve the approximation ratio of the buyer-offering with reserve $\reserve$ mechanism when the buyer value is $v_b$, since for these seller values every mechanism achieves the optimal allocation. It also does not improve the approximation ratio when $v_b =\reserve$, as it only reduces $E[v_s]$. The resulting distribution is $\worsedis$. 
\end{proof}

With Claim~\ref{claim:price-equal-reserve} in hand, we analyze the approximation ratio when the buyer's value is $v_b$. We start with a bound on the expected value of the seller when its value is at least $\reserve$ (Claim~\ref{claim:bound-on-seller-val}). We essentially show that the worst case for our analysis, i.e., the value of the buyer where the smallest approximation ratio is achieved, is when the seller's distribution is an equal profit distribution for a buyer with value $v_b$. In this distribution, all offers $\price > \reserve$ with $\worsedis(\price)< 1$ yield the same profit for a buyer with value $v_b$.

\begin{claim}\label{claim:bound-on-seller-val}
    $\mathbb E_{v_s\sim \worsedis}[\indicator_{[ \reserve < v_s \leq v_b]}\cdot v_s]
    \geq \worsedis(\reserve)(v_b-\reserve)\lnn{\worsedis(\reserve)}
    + v_b\cdot(1 - \worsedis(\reserve)).$
\end{claim}   

\begin{proof}[Proof of Claim~\ref{claim:bound-on-seller-val}]   
    Let $\worsedispdf$ be the probability density function of the distribution $\worsedis$, and let
    $\utility = (v_b-\reserve)\cdot\worsedis(\reserve)$ be the profit of the buyer with value $v_b$
    from making an offer $\reserve$. Then
        \begin{equation*}
        \begin{aligned}[b]
        \mathbb E_{v_s\sim \worsedis}[\indicator_{[ \reserve < v_s \leq v_b]}\cdot v_s]
        & = \int_{\reserve}^{v_b} v_s\cdot \worsedispdf(v_s)\, \d v_s \\
        & \underbrace{=}_{\text{integration by parts}}
        \eval{v_s\worsedis(v_s)}_\reserve^{v_b}
        -\int_\reserve^{v_b} \worsedis(v_s)\, \d v_s \\
        & = v_b -\reserve\cdot \worsedis(\reserve)
        -\int_\reserve^{v_b} \worsedis(v_s)\, \d v_s \\
        & \geq v_b -\reserve\cdot \worsedis(\reserve)
        -\int_\reserve^{v_b-\utility} \frac{\utility}{v_b-v_s}\, \d v_s
        -\int_{v_b-\utility}^{v_b} 1\, \d v_s \\
        & = v_b -\reserve\cdot \worsedis(\reserve)
        -\utility
        +\eval{\utility\cdot \lnn{(v_b-v_s)}}_{\reserve}^{v_b-\utility} \\
        & = v_b -\reserve\cdot \worsedis(\reserve)
        -\utility
        +\utility\cdot \lnn{\frac{\utility}{v_b-\reserve}} \\
        & = v_b -\reserve\cdot \worsedis(\reserve)
        -\utility
        +\worsedis(\reserve)(v_b-\reserve)\lnn{\worsedis(\reserve)} \\
        & = \worsedis(\reserve)(v_b-\reserve)\lnn{\worsedis(\reserve)}
        + v_b\cdot(1 - \worsedis(\reserve)).
        \end{aligned}
    \end{equation*}
\end{proof}

    Now, the approximation ratio of the buyer-offering with reserve $\reserve$ when the value of the buyer is $v_b$ is at least:

\begin{equation}\label{ineq:app-at-b}
\begin{aligned}[t]
    \frac{\bor{\reserve}{\worsedis}(v_b)}{\opt{\worsedis}(v_b)}
    & = \frac{v_b\cdot\worsedis(\reserve)
    + \mathbb E_{v_s\sim \worsedis}[\indicator_{[ \reserve < v_s \leq v_b]}\cdot v_s]}{v_b}  \geq \frac{v_b+ (v_b-\reserve)\cdot\worsedis(\reserve)\lnn{\worsedis(\reserve)}}{v_b}  \\
    & = 1+\worsedis(\reserve)\Bigl(1-\frac{\reserve}{v_b}\Bigr)\lnn{\worsedis(\reserve)} .
\end{aligned}
\end{equation}

    We are almost ready to state our optimization problem, but first we need to rewrite the constraint about the approximation ratio of a buyer with value $\reserve$ being at most $\appbelow$.

    \begin{claim}\label{claim:app-at-r}
        Suppose that the item is never traded for all possible values of the seller and that the value of the buyer is $\reserve$. Let $\appbelow$ be an upper bound on the approximation ratio in this case. Then:
        $$
        \frac{\appbelow}{1-\appbelow}\cdot\frac{\reserve}{v_b}\cdot\worsedis(\reserve) \geq \worsedis(\reserve)\cdot(1-\frac{\reserve}{v_b})\cdot\lnn{\worsedis(\reserve)} + 1 -\worsedis(\reserve)
        $$
    \end{claim}

    \begin{proof}[Proof of Claim~\ref{claim:app-at-r}]
        The approximation ratio of a buyer with value $\reserve$ is at most $\appbelow$, hence:
    \begin{align*}
        \appbelow & \geq   \frac{\mathbb E_{v_s\sim \worsedis}[v_s]}{\reserve\cdot\worsedis(\reserve)+\mathbb E_{v_s\sim \worsedis}[v_s \cdot \indicator_{[v_s \geq \reserve]}]} =  1+  \frac{\mathbb E_{v_s\sim \worsedis}[v_s \cdot \indicator_{[v_s < \reserve]}]-\reserve\cdot\worsedis(\reserve)}{\reserve\cdot\worsedis(\reserve)+ \mathbb E_{v_s\sim \worsedis}[v_s \cdot \indicator_{[v_s \geq \reserve]}]} \\
        & \iff (1-\appbelow) \leq  \frac{\reserve\cdot\worsedis(\reserve) - \mathbb E_{v_s\sim \worsedis}[v_s \cdot \indicator_{[v_s < \reserve]}]}{\reserve\cdot\worsedis(\reserve)+ \mathbb E_{v_s\sim \worsedis}[v_s \cdot \indicator_{[v_s \geq \reserve]}]}\\
        & \iff (1-\appbelow)\cdot \mathbb E_{v_s\sim \worsedis}[v_s \cdot \indicator_{[v_s \geq \reserve]}] + \mathbb E_{v_s\sim \worsedis}[v_s \cdot \indicator_{[v_s < \reserve]}] \leq \appbelow\cdot\reserve\cdot\worsedis(\reserve)
    \end{align*}
    Using Claim~\ref{claim:bound-on-seller-val} we get that:
    \begin{align*}
    \appbelow\cdot\reserve\cdot\worsedis(\reserve) & \geq (1-\appbelow)\cdot \mathbb E_{v_s\sim \worsedis}[v_s \cdot \indicator_{[v_s \geq \reserve]}] \geq (1-\appbelow)\cdot \left( \worsedis(\reserve)({v_b}-\reserve)\lnn{\worsedis(\reserve)} + {v_b}\cdot(1 - \worsedis(\reserve)) \right) \nonumber\\
    & \iff \frac{\appbelow}{1-\appbelow}\cdot\frac{\reserve}{v_b}\cdot\worsedis(\reserve) \geq \worsedis(\reserve)\cdot(1-\frac{\reserve}{v_b})\cdot\lnn{\worsedis(\reserve)} + 1 -\worsedis(\reserve)
    \end{align*}

    \end{proof}
    
Let $x= \worsedis(\reserve)$, $\ell = \frac{\reserve}{v_b}$. By Inequality~\eqref{ineq:app-at-b} and Claim~\ref{claim:app-at-r}, we get that the solution to the following minimization problem bounds $\appabove$ from below.

\begin{equation}\label{program:min-app}
\begin{aligned}
 & \min_{1-\appbelow \leq \distribution(\reserve) \leq x \leq 1, 0 < l\leq 1}   \Phi(x,l) := 1+x(1-\ell)\ln{x} \\
   & \hspace{10mm} \text{ s.t. }  \frac{\appbelow}{1-\appbelow}\cdot\ell x \geq 1-x + x(1-\ell)\ln{x} \text{ and } \ell \geq \lambda_{v_b}\\
\end{aligned}
\end{equation}

To see why $x\geq 1-\appbelow$, recall that the approximation ratio when $v_b = \reserve$ and the item is never traded is at most $\appbelow$. Therefore:

\begin{align*}
    \appbelow & \geq \frac{\mathbb E_{v_s \sim \worsedis}[v_s]}{\opt{\worsedis}(\reserve)} \geq \frac{\mathbb E_{v_s\sim \distribution}[\indicator_{[ v_s > \reserve]}\cdot v_s]}{\reserve\worsedis(\reserve) + \mathbb E_{v_s\sim \distribution}[\indicator_{[ v_s > \reserve]}\cdot v_s]} = 1-\frac{\reserve\worsedis(\reserve)}{\reserve\worsedis(\reserve) + \mathbb E_{v_s\sim \distribution}[\indicator_{[ v_s > \reserve]}\cdot v_s]}\\
    & \geq 1-\frac{\reserve\worsedis(\reserve)}{\reserve\worsedis(\reserve) + \reserve(1-\worsedis(\reserve))} = 1-\worsedis(\reserve)
\end{align*}

Hence, $x=\worsedis(\reserve) \geq 1-\appbelow$. Next, we analyze this optimization problem and show that its minimal value is obtained on the boundaries of the constraint on $\ell$ (Claim~\ref{claim:properties-of-opt-problem}). 

\begin{claim}\label{claim:properties-of-opt-problem}
    The optimization problem in \eqref{program:min-app} has a minimal value and satisfies these conditions:
    \begin{enumerate}
        \item The constraint $ \frac{\appbelow}{1-\appbelow}\cdot\ell x \geq 1-x + x(1-\ell)\ln{x}$ is equivalent to $\ell\geq  \frac{1-x+x\ln{x}}{\frac{x\cdot\appbelow}{1-\appbelow}+x\ln{x}} $ for $x \in [1-\appbelow, 1]$.
        \item The minimum is obtained when one of these two constraints on $\ell$ are met with equality:
        \begin{itemize}
            \item $\ell \geq \lambda_{v_b}$
            \item $\ell\geq  \frac{1-x+x\ln{x}}{\frac{x\cdot\appbelow}{1-\appbelow}+x\ln{x}} $
        \end{itemize}
    \end{enumerate}
\end{claim}

\begin{proof}[Proof of Claim~\ref{claim:properties-of-opt-problem}]
First, we analyze the constraint in the optimization problem~\eqref{program:min-app}:

\begin{align*}
    \ell \left( x\cdot\frac{\appbelow}{1-\appbelow} +x\ln{x}\right) \geq 1-x+x\ln{x} \iff \ell \geq \frac{1-x+x\ln{x}}{x\cdot\frac{\appbelow}{1-\appbelow} +x\ln{x}} \text{ when } x\cdot\frac{\appbelow}{1-\appbelow} +x\ln{x} > 0
\end{align*}

We show that for $\appbelow \in (0,1)$ it must be that $\frac{\appbelow}{1-\appbelow} +\ln{x} > 0$, which will imply the first part of this claim as $x>0$. Since $\ln{x}$ is a monotone increasing function, our condition holds if $x>e^{\frac{-\appbelow}{1-\appbelow}}$. 
We observe this by recalling that $x\geq 1-\appbelow$, then since $1-\appbelow-e^{\frac{-\appbelow}{1-\appbelow}} >0 $ for $\appbelow \in (0,1)$ we would get that $x> e^{\frac{-\appbelow}{1-\appbelow}}$ for $x$ in our region. 
    
For the second part of this claim, note that our objective function $\Phi(x,l) := 1+x(1-\ell)\ln{x}$ is continuous in $x \in [1-\appbelow, 1], \ell \in [0,1]$, since it is a sum and product of continuous functions. Therefore, it attains a minimum on the compact set $ [1-\appbelow, 1] \times [0,1]$.  
Note that $ \frac{\partial \Phi}{\partial \ell} = -x\ln{x}$, and so the partial derivative of $\Phi$ with respect to $\ell$ is non-zero in our domain $x\in [1-\alpha,1]$ (note that if $x=1$ we get the optimal ratio of $1$ so we can assume that $x<1$). Hence, any extreme points must be on the boundary, i.e., $\ell =  \frac{1-x+x\ln{x}}{\frac{x\cdot\appbelow}{1-\appbelow}+x\ln{x}} $, or $x = 1-\appbelow$ or $\ell = \lambda_{v_b}$.
To conclude the proof, we need to rule out the case that $x=1-\appbelow$. If $x=1-\appbelow$, the constraint on $\ell$ becomes:

\begin{align*}
    \ell \geq \frac{1-x+x\ln{x}}{\frac{x\cdot\appbelow}{1-\appbelow}+x\ln{x}} = \frac{1-(1-\appbelow)+(1-\appbelow)\lnn{1-\appbelow}}{\frac{(1-\appbelow)\cdot\appbelow}{1-\appbelow}+(1-\appbelow)\lnn{1-\appbelow}} =1
\end{align*}

Recall that $\ell \leq 1$. Therefore, $\ell =1$ and our objective becomes $\Phi=1$ which is a maximum point. We conclude that $x=1-\appbelow$ cannot be a minimum point.
\end{proof}

\end{proof}

\subsection{The Neighborhood Lemma}

\begin{lemma}[Neighborhood Lemma]\label{lemma:distance}
    Let $0 < r_1 < r_2 $ be such that 
    $$\frac{\mathbb E_{v_s \sim \distribution}[v_s]}{\opt{\distribution}(r_1)} = \alpha_1, \quad 1>\frac{\mathbb E_{v_s \sim \distribution}[v_s]}{\opt{\distribution}(r_2)} \geq  \alpha_2$$
    With $\alpha_1 >0$, then, 
    $$
    \frac{r_1}{r_2} \geq \frac{\alpha_2(1-\alpha_1)}{\alpha_1(1-\alpha_2)}
    $$
\end{lemma}

\begin{proof}
    Let $G(v_b) = \int_{0}^{v_b} (v_b -v_s)\distributionpdf(v_s) \d v_s $ be the optimal gains-from-trade when the buyer's value is $v_b$. Let $A(v_b)$ be the approximation ratio when the buyer's value is $v_b$ and the item is never traded, i.e.,  $A(v_b) = \frac{\mathbb E_{v_s \sim \distribution}[v_s]}{v_b\cdot\distribution(v_b) + \mathbb E_{v_s \sim \distribution}[v_s\cdot \indicator_{[v_s > v_b]}]}$.
    We get:
    \begin{equation}
        \begin{aligned}
                A(v_b) & = \frac{\mathbb E_{v_s \sim \distribution}[v_s]}{v_b\cdot\distribution(v_b) + \mathbb E_{v_s \sim \distribution}[v_s] -\mathbb E_{v_s \sim \distribution}[v_s\cdot \indicator_{[v_s \leq v_b]}]} = \frac{\mathbb E_{v_s \sim \distribution}[v_s]}{\mathbb E_{v_s \sim \distribution}[v_s] + G(v_b)} \\
                & \underbrace{\iff}_{\mathbb E_{v_s \sim \distribution}[v_s] > 0} G(v_b) = \mathbb E_{v_s \sim \distribution}[v_s]\cdot\frac{1-A(v_b)}{A(v_b)} \nonumber
        \end{aligned}
    \end{equation}
    Where we used $\mathbb E_{v_s \sim \distribution}[v_s] > 0$ as otherwise this expectation is $0$ and thus $\alpha_1 = \alpha_2 =0$, which violates the conditions of the lemma. We observe that $G(\cdot)$ is convex:
    \begin{claim}\label{claim:g-convex}
        The function $G(\cdot)$ is convex in $\mathbb R_{\geq 0}$.
    \end{claim}
    \begin{proof}[Proof of Claim~\ref{claim:g-convex}]
        Recall the definition $G(v_b)=\int_0^{v_b} (v_b-v_s)\distributionpdf(v_s) \d v_s$.
        \begin{align*}
        G(v_b) & = v_b \cdot \distribution(v_b) -\int_{0}^{v_b} v_s\distributionpdf(v_s) \d v_s \underbrace{=}_{\text{integration by parts}} v_b \cdot \distribution(v_b) -\left(\eval{v_s\cdot\distribution(v_s)}_{0}^{v_b} -\int_{0}^{v_b} \distribution(v_s) \d v_s\right) \\
        & = \int_{0}^{v_b} \distribution(v_s) \d v_s
        \end{align*}
        By the fundamental theorem of calculus, we have $G'(v_b) = \distribution(v_b)$. Since $\distribution(v_b)$ is monotonically non-decreasing, we get that $G(\cdot)$ is convex.   
    \end{proof}

    Let $\lambda = \frac{r_1}{r_2}$. $\lambda < 1 $ since $r_1 < r_2$. By the convexity of $G$, we get $G(r_1) = G((1-\lambda)\cdot 0 + \lambda\cdot r_2) \leq (1-\lambda)\cdot G(0)+ \lambda\cdot G(r_2) = \lambda\cdot G(r_2)$, where we used $G(0)=0$. Hence,
    \begin{equation*}
        \frac{r_1}{r_2} \geq \frac{G(r_1)}{G(r_2)} = \frac{\left(1-A(r_1)\right)A(r_2)}{A(r_1)\left(1-A(r_2)\right)} \geq \frac{\alpha_2(1-\alpha_1)}{\alpha_1(1-\alpha_2)}
    \end{equation*}
    In the first inequality we use that $G(r_2) >0$, since $A(r_2)<1$. The second inequality follows because  $\frac{A(r_2)}{1-A(r_2)}$ is a monotone increasing function in $A(r_2)$ for $A(r_2) \in [0,1]$.
    \end{proof}

%% file: seller-only.tex
\section{Mechanisms with No Information about the Buyer}\label{sec-seller-only}

We now construct mechanisms that guarantee a good fraction of the optimal social welfare for every distribution of the buyer $\mathcal F_b$. Specifically, for every possible value of the buyer $v_b$, the approximation ratio (conditioned on the value of the buyer being $v_b$) is guaranteed to be at least $0.71$. In Section \ref{sec-main-result}, we guarantee a larger fraction of the optimal welfare, but this improvement comes at the cost that this fraction is guaranteed only in expectation, also over the distribution of the buyer. That is, for some values $v_b$ of the buyer, the expected fraction of the optimal welfare extracted might be lower than the guarantee (but for other values, it might be higher). 

We note that, for fixed-price mechanisms, it is known that the approximation ratios achievable by mechanisms that use only the distribution of one player are the same as those achievable by mechanisms that might use both distributions \cite{CW-STOC}. However, mechanisms that use the distribution of only one player and achieve the state-of-the-art approximation ratios must be randomized, as it is known that deterministic fixed-price mechanisms that use only the distribution of one player cannot guarantee more than half of the optimal social welfare \cite{BD14}. Thus, the following theorem demonstrates another aspect in which buyer-offering mechanisms beat fixed-price mechanisms:

\begin{theorem}
    For every distribution $\mathcal F_s$ of the seller, there is a reserve $r$, such that for every $v_b$, the fraction of the welfare that the buyer offering mechanism with reserve $r$ extracts (conditioned on the value of the buyer being $v_b$) is at least $0.71004$, regardless of the distribution $\mathcal F_b$.
\end{theorem}

The proof is an application of the tradeoff lemma, by equalizing the approximation ratio guaranteed for values below and above the reserve.
\begin{proof}
Fix a distribution $\mathcal F_s$ over the values of the seller. We consider the buyer-offering mechanism with a reserve $r$, to be chosen as a function of $\distribution$ only.

For any buyer value $v_b< r$, the reserve prevents trade, and therefore the mechanism's welfare equals $\E_{v_s\sim\distribution}[v_s]$.  The optimal welfare conditioned on $v_b$ is
\[
\opt{\distribution}(v_b)= v_b\cdot \distribution(v_b) \;+\; \E_{v_s\sim\distribution}\big[v_s\cdot \indicator_{[v_s>v_b]}\big].
\]
In particular, $\opt{\distribution}(v_b)$ is non-decreasing in $v_b$, and hence for every $v_b\le r$,
\[
\frac{\E_{v_s\sim\distribution}[v_s]}{\opt{\distribution}(v_b)}
\;\;\ge\;\;
\frac{\E_{v_s\sim\distribution}[v_s]}{\opt{\distribution}(r)}
\]
Fix a target value $\alpha\in(0,1)$ and choose $r$ so that
\begin{equation}\label{eq:choose-r}
\frac{\E_{v_s\sim\distribution}[v_s]}{\opt{\distribution}(r)}=\alpha
\end{equation}
Existence of such $\alpha$ follows from continuity: at $v_b=0$ the ratio equals $1$, and as $v_b$ increases the denominator $\opt{\distribution}(v_b)$ increases, driving the ratio down. Then, by the inequality above, the welfare fraction guaranteed by the mechanism is at least $\alpha$ for all $v_b\le r$.

Now consider any buyer value $v_b>r$. By construction, \eqref{eq:choose-r} exactly matches the premise of Lemma~\ref{lemma:assymetric} with $\appbelow=\alpha$ and $\reserve=r$. Applying Corollary~\ref{cor:threshold-lemma-no-bound-on-v-b}, it follows that for every $v_b>r$, the welfare fraction guaranteed by the buyer-offering mechanism with reserve $r$ is at least
\begin{equation*}\label{eq:phi-def}
\phi(\alpha)\;=\;
\min_{1-\alpha \leq x\leq 1}
\left(
1+x\ln x
-
\frac{\ln x \left(1-x+x\ln x\right)}{\ln x +\frac{\alpha}{1-\alpha}}
\right).
\end{equation*}

Combining the two steps, for the reserve $r$ chosen according to \eqref{eq:choose-r}, the mechanism guarantees, for every buyer value $v_b$, a welfare fraction of at least $\min\{\alpha,\;\phi(\alpha)\}$.
We therefore choose $\alpha$ to maximize this expression, which occurs at $\alpha=\phi(\alpha)$.
We now show that $\alpha \ge 0.71004$. To see this, fix $\alpha=0.71004$ and consider the function
\[
\phi(0.71004)
=
\min_{1-0.71004 \leq x\leq 1}
\left(
1+x\ln x
-
\frac{\ln x \left(1-x+x\ln x\right)}{\ln x +\frac{0.71004}{1-0.71004}}
\right)
\]
Let
\[
h(x)\;=\;
1+x\ln x
-
\frac{\ln x \left(1-x+x\ln x\right)}{\ln x +\frac{0.71004}{1-0.71004}}
\]
On the interval $x\in[1-0.71004,1]=[0.28996,1]$, the function $h(x)$ is smooth. On the interval $x\in[1-\alpha,1]$, the function $h(x)$ is continuously differentiable. Its derivative is given by

\begin{equation}\label{equ:deriviative-h}
    h'(x) = 1+\ln{x} -\frac{(\ln^2{x}+c)}{\ln{x}+ c} + \frac{\left(\frac{\alpha}{1-\alpha}\right)^2-\frac{\alpha}{x(1-\alpha)}+\frac{\alpha}{1-\alpha}}{\left(\ln{x}+ c\right)^2}
\end{equation}

A direct numerical inspection shows that $h'(x)$ has a unique zero in the interval
$x\in[1-\alpha,1]$, occurring at $x^\star \approx 0.539006$.
Moreover, $h'(x)<0$ for $x<x^\star$ and $h'(x)>0$ for $x>x^\star$, implying that
$x^\star$ is the unique minimizer of $h(x)$ on this interval.
Evaluating the function at this point yields $h(x^\star)\approx 0.710047$.
Thus,
\[
\phi(0.71004)=\min_{x\in[0.28996,1]} h(x) = h(x^\star) \approx 0.710047 ,
\]
It follows that the buyer-offering mechanism with reserve $r$ guarantees a welfare fraction of at least $0.71004$ for every buyer value $v_b$.
\end{proof}

We also show that the approximation ratio guaranteed by our mechanism is very close to be optimal for this class of mechanisms.

\begin{theorem}
There exists a distribution of the seller $\mathcal F_s$ such that for every reserve $r$, there is a buyer value $v_b$ for which the fraction of the optimal social welfare guaranteed by the buyer-offering mechanism with reserve $r$ (conditioned on $v_b$) is at most $0.7159$.
\end{theorem}

\begin{proof}
We construct a specific distribution $\mathcal F_s$ and analyze the performance of the buyer-offering mechanism with an arbitrary reserve $r$.

Let $x=0.455$ and consider the equal-profit distribution supported on $[0,1-x]$, defined by the cumulative distribution function
\[
F_s(t)=\frac{x}{1-t}\qquad\text{for }t\in[0,1-x].
\]
A direct calculation shows that the expected value of the seller under this distribution is
\[
\E[v_s]=\int_{0}^{1-x} (1-F_s(t))\,dt
=1-x+x\ln x
\approx 0.1867.
\]

For every reserve $r$, we consider two buyer values and show that at least one of them induces an approximation ratio of at most $0.7159$.

When the buyer's value is $1$, the buyer-offering mechanism with reserve $r$ makes the offer $r$. Trade occurs if and only if $v_s \le r$. In this case, the welfare is $1$; otherwise, the welfare equals $v_s$. Thus, the expected welfare of the mechanism is
\[
F_s(r)\cdot 1 + \int_{r}^{1-x} t \, dF_s(t).
\]
Since the optimal welfare is $1$, the welfare fraction guaranteed by the mechanism in this case equals this expression.

Fix $\varepsilon>0$ arbitrarily small. When $v_b=r-\varepsilon$, the reserve prevents trade, and the welfare obtained by the mechanism is simply $\E[v_s]$. The optimal welfare, conditioned on this buyer value, is
\[
\E[\max\{v_s,v_b\}]
=(r-\varepsilon)\cdot F_s(r-\varepsilon)
+\int_{r-\varepsilon}^{1-x} t\, dF_s(t).
\]
Letting $\varepsilon\to 0$, the welfare fraction in this case converges to
\[
\frac{\E[v_s]} {rF_s(r)+\int_{r}^{1-x} t\, dF_s(t)}.
\]

For every $r\in[0,1-x]$,
\[
F_s(r)+\int_{r}^{1-x} t\, dF_s(t)
=\frac{x}{1-r} + (1-x)-\frac{rx}{1-r}+x\ln\!\left(\frac{x}{1-r}\right),
\]
while
\[
rF_s(r)+\int_{r}^{1-x} t\, dF_s(t)
=1-x+x\ln\!\left(\frac{x}{1-r}\right).
\]
Define $W_1,W_2$ be the approximation ratios when the buyer's value is $1$ and when it approaches $r$, respectively, when the reserve is $r$:
\[
W_1(r)\;=\;F_s(r)+\int_{r}^{1-x} t\, dF_s(t)
\qquad\text{and}\qquad
W_2(r)\;=\;\frac{\E[v_s]}{rF_s(r)+\int_{r}^{1-x} t\, dF_s(t)},
\]
$W_1(r)$ is increasing and $W_2(r)$ is decreasing on $[0,1-x]$. Therefore, the function
\[
m(r)\;=\;\min\{W_1(r),W_2(r)\}
\]
is maximized at a point where $W_1(r)=W_2(r)$ (when such a point exists). Since $W_1(0)<W_2(0)$ and
$W_1(1-x)=1>W_2(1-x)$, continuity implies that there exists $r^\star\in(0,1-x)$ such that $W_1(r^\star)=W_2(r^\star)$.
At this point,
\[
m(r^\star)=W_1(r^\star)=W_2(r^\star)\approx 0.7158267.
\]
In particular, $m(r^\star)\le 0.7159$. Since $m(r)$ is maximized at $r^\star$, we have $m(r)\le 0.7159$ for every reserve
price $r$. In other words, for every reserve there is a value of the buyer for which the approximation ratio is at most $0.7159$.
\end{proof}

Note that this bound does not exclude the possibility that a distribution over buyer-offering mechanisms will beat this bound, even when the distribution of the buyer is not known.

%% file: both-sides.tex
\section{The Main Result: A $\upperbound$-Approximation Mechanism}\label{sec-main-result}

In this section, we prove that there exists a buyer-offering mechanism with a reserve $\reserve$ that provides a $\upperbound$-approximation ratio for any pair of agent value distributions.

\begin{theorem}\label{thm-main}
    {For every two distributions $\mathcal F_s, \mathcal F_b$ of the seller and the buyer, there is a reserve $r$, such that the expected welfare of the buyer-offering mechanism with reserve $r$ is at least a fraction of $\upperbound$ of the welfare of the optimal social welfare.}
\end{theorem}

We consider a finite family of buyer-offering mechanisms with different reserve prices and argue that at least one of them must perform well against any pair of buyer and seller value distributions. This is the high-level idea of the argument. We start with a decreasing sequence  
\[
\sigma = {\appbelow}_1 \ge \dots \ge {\appbelow}_n
\]
of $n$ values in $[0,1]$, where for each ${\appbelow}_i$ there exists a buyer value in the support such that the approximation ratio of the mechanism that never trades equals ${\appbelow}_i$. 

Given a sequence $\sigma$, we partition the buyer’s value support into $n+1$ intervals $I_1, \dots, I_{n+1}$ according to the values in $\sigma$, such that each value in the $j$th interval has an approximation ratio of at least ${\appbelow}_j$ and at most ${\appbelow}_{j-1}$ when the item is never traded.  
The same sequence $\sigma$ also determines the reserve prices we consider, $r_1, \dots, r_n$: each reserve corresponds to a buyer value at which the approximation ratio of the no-trade mechanism equals ${\appbelow}_i$.

We then use the bounds from the Tradeoff Lemma and the Neighborhood Lemma to derive worst-case lower bounds on the approximation ratio of the buyer-offering mechanism with each reserve for buyer values in each interval. These bounds define the entries of a coefficient matrix. That is, we define a matrix $A$, where each row $i$ corresponds to a buyer-offering mechanism with reserve $r_i$, and each entry $A_{i,j}$ provides a lower bound on the approximation ratio when the buyer’s value falls within interval $I_j$.

Each buyer distribution induces weights over the intervals, with the overall approximation ratio given by a weighted sum of the interval-wise guarantees. The worst-case choice of weights is captured by a linear program, whose optimal value lower-bounds the approximation ratio of the best reserve price for any buyer distribution.

Finally, we construct a specific sequence $\sigma$ and show that the value of the resulting linear program is at least $\upperbound$. This is done by constructing a feasible dual solution, which certifies the $\upperbound$ approximation guarantee.

\begin{proposition}\label{prop:lp-bound-given-sigma}
Let $\sigma = ({\appbelow}_1,\dots,{\appbelow}_n)$ be a decreasing sequence of values in $[0,1]$. Then the optimal value of the linear program~\eqref{program:lp-primal}, defined using the coefficient matrix $A_\sigma$ constructed as in~\eqref{def:set-coeff-matrix}, provides a lower bound on the approximation ratio of the best buyer-offering mechanism with reserve.
\end{proposition}

The proof of the proposition consists of several steps and is provided next. After proving the proposition, we construct a specific sequence $\sigma$ (with $n=19$), and show\footnote{Of course, it is possible to simply put the coefficient matrix in an LP solver and get a solution. However, for completeness, we provide a proof of this fact by considering the dual.} that the value of the linear program with these specific variables is at least $\upperbound$ (Appendix~\ref{sec:exact-lp-values}). Thus, this also establishes Theorem \ref{thm-main}. Of course, we expect that adding more values to the sequence will result in an improved bound.

\subsection*{Proof of Proposition \ref{prop:lp-bound-given-sigma}}

\subsubsection*{Step 1: A Sequence $\sigma$}\label{sec:sequence-from-seller-dis}
Fix a distribution $\distribution$ of the seller. Define $\appbelow(v_b)$ to be the approximation ratio when the buyer's value is $v_b$ and the mechanism never trades: 
$$\appbelow(v_b) = \frac{\mathbb E_{v_s \sim \distribution}[v_s]}{v_b \cdot \distribution(v_b) + \mathbb E_{v_s \sim \distribution}[v_s\cdot\indicator_{[v_s >v_b]}]}
$$
Note that $\appbelow(0)=1$ and that $\appbelow(v_b)$ is monotone non-increasing. Thus, its minimal value is attained at the highest value in the support, that is, $v_b=1$. Let ${\appbelow}_{\min}$ denote this value.
Observe that the function $\appbelow(\cdot)$ is continuous even if  $\distribution$ is discrete. To see this, consider a discrete seller distribution $\distdis$, and let $\suppseller = \{ s_1 < \dots < s_n\} $ be its support \footnote{Note that if only part of the distribution is discrete, then there will be a value $s_{i}$ with positive probability, and the distribution will become continuous on some interval $(s_{i-1}, s_{i})$. In this case, we have $\mathcal{F}(s_i) = \lim_{s \to s_i^{-}}\mathcal{F}(s) + \Pr[v_s = s_i]$, and the function $\appbelow(v_b)$ for $v_b \in (s_{i-1}, s_i)$ becomes
\[
\appbelow(v_b) = \frac{\mathbb{E}_{v_s \sim \mathcal{F}}[v_s]}{v_b \cdot \lim_{s\to v_b^-}\mathcal{F}(v_b) + s_i \cdot\Pr[v_s =s_i] +  \lim_{s \to s_i^-} \int_{v_b}^{s} v_s \cdot\distributionpdf(v_s) \d v_s  +\mathbb{E}_{v_s \sim \mathcal{F}}[v_s \cdot \mathbf{1}_{[v_s > s_i]}]},
\]
Note that when $v_b \to s_i^-$ we get: 
$$
\lim _{v_b \to s_i^-} \appbelow(v_b) =  \frac{\mathbb{E}_{v_s \sim \mathcal{F}}[v_s]}{v_b \cdot \lim_{s\to v_b^-}\mathcal{F}(v_b) + s_i \cdot\Pr[v_s =s_i]  +\mathbb{E}_{v_s \sim \mathcal{F}}[v_s \cdot \mathbf{1}_{[v_s > s_i]}]}  =\frac{\mathbb{E}_{v_s \sim \mathcal{F}}[v_s]}{s_i \cdot \mathcal{F}(s_i) +\mathbb{E}_{v_s \sim \mathcal{F}}[v_s \cdot \mathbf{1}_{[v_s > s_i]}]} =\appbelow(s_i),
$$
and thus the function is still continuous.
}.
Fix an index $i \in \{2, \dots, n\}$ and consider buyer values $v_b \in [s_{i-1}, s_i]$. For such values, the function $\appbelow(\cdot)$ can be written as:
$$\appbelow(v_b) = 
    \frac{\mathbb E_{v_s \sim \distribution}[v_s]}{v_b \cdot \distribution(s_{i-1}) + s_i\cdot(\distribution(s_i)-\distribution(s_{i-1})) + \mathbb E_{v_s \sim \distribution}[v_s\cdot\indicator_{[v_s >s_i]}]} 
 $$
All terms in the denominator except for $v_b$ are constant on the interval $[s_{i-1}, s_i]$, and the dependence on $v_b$ is linear. Consequently, $\appbelow(v_b)$ is continuous on each interval $[s_{i-1}, s_i]$. Since these intervals cover the entire support of buyer values, the function $\appbelow(\cdot)$ is continuous everywhere.

By the intermediate value theorem, for every $\alpha \in [{\appbelow}_{\min}, 1]$, there is $v_b \in \suppbuyer$ such that $\appbelow(v_b) = \alpha$.  
The case that ${\appbelow}_{\min} > 0$ is easier and will be discussed in Remark~\ref{remark:exact-seq}. For now, we assume that ${\appbelow}_{\min}=0$.
We may therefore choose any $n$ decreasing values in $[0,1]$ and denote them by
\[
\sigma = ({\appbelow}_1 > {\appbelow}_2 > \cdots > {\appbelow}_n).
\]
For each ${\appbelow}_i$, let $r_i$ be a buyer value for which $\appbelow(r_i)={\appbelow}_i$. Since $\appbelow(\cdot)$ is monotone, we have
\[
0 \leq r_1 < r_2 < \cdots < r_n,
\]
and these values induce a partition of the buyer’s value support into $n+1$ intervals
\[
I_1= [0,r_1],\; I_2=(r_1,r_2],\; \ldots,\; I_n=(r_{n-1},r_n],\; I_{n+1}=(r_n,1].
\]

By construction, for every buyer value $v_b \in I_j$, the approximation ratio of the mechanism that never trades satisfies
\[
{\appbelow}_j \le {\appbelow}(v_b) \le {\appbelow}_{j-1}.
\]


Throughout the proof, we may refer to the seller and buyer distributions, but the analysis does not require their exact specification. The actual reserve values are irrelevant; it suffices that they induce a partition of the buyer support into $n+1$ intervals and yield interval-wise approximation guarantees.

\subsubsection*{Step 2: The Coefficient Matrix}\label{sec:coeff-matrix}
Recall that for every interval $I_j$ we have a lower bound on the approximation ratio when the item is never traded and the buyer’s value lies in $I_j$. We now derive analogous bounds for buyer-offering mechanisms with reserves drawn from
\[
R := \{r_1,\dots,r_n\}.
\]

Fix a reserve $r_i$. Using the Tradeoff Lemma, we obtain a lower bound on the approximation ratio of the buyer-offering mechanism with reserve $r_i$ for all buyer values $v_b \in I_j$ with $j>i$. Intuitively, these are precisely the buyer values for which trade may occur under reserve $r_i$.

We associate each ${\appbelow}_i \in \sigma$ with a corresponding value ${\appabove}_i$, defined as the worst-case approximation ratio guaranteed by the buyer-offering mechanism with reserve $r_i$ over all buyer values in intervals $I_j$ with $j>i$. The value ${\appabove}_i$ is given by the solution to the optimization problem specified in Corollary~\ref{cor:threshold-lemma-no-bound-on-v-b}.

Combining these bounds across all reserves and intervals defines the coefficient matrix $A$, where each row $i$ corresponds to reserve $r_i$, and each entry $A_{i,j}$ gives a lower bound on the approximation ratio when the buyer’s value lies in interval $I_j$. 

We further exploit the structure of the reserve prices to obtain tighter bounds when the reserve is $r_i$ and the buyer’s value lies in an interval $I_j$ with $j>i$ but close to $i$. These refined bounds follow from the Neighborhood Lemma, which captures the improvement in the approximation ratio when the reserve and buyer value are nearby. This refinement is incorporated into the definition of the coefficient matrix in~\eqref{def:set-coeff-matrix}.

Formally, we construct a matrix $A \in \mathbb{R}_{\geq 0}^{n \times (n+1)}$, where the $i$th row corresponds to the buyer-offering mechanism with reserve $r_i$, and the $j$th column corresponds to buyer values lying in the interval $I_j$. For every $i \in [n]$ and $j \in [n+1]$, the approximation ratio of the buyer-offering mechanism with reserve $r_i$ is at least $A_{i,j}$ whenever the buyer’s value $v_b$ lies in $I_j$.
We construct the entries of the matrix $A$ as follows.
For indices $i<j\leq n$, define
\begin{equation*}
    \lambda_{i,j}
    \;=\;
    \frac{{\appbelow}_j(1-{\appbelow}_i)}
         {{\appbelow}_i(1-{\appbelow}_j)} .
\end{equation*}
For $i<j \leq n$, let $\alpha_{i,j}$ denote the value of $\appabove(\appabove, \lambda_{v_b})$ in
Corollary~\ref{cor:threshold-lemma-bound-on-v-b} with parameters
$\lambda_{v_b}=\lambda_{i,j}$ and $\appbelow={\appbelow}_i$.
We then define
\begin{equation}\label{def:set-coeff-matrix}
    A_{i,j}
    =
    \begin{cases}
        \max\{\alpha_{i,j},\,{\appabove}_i\}, & i<j \leq n,\\
        {\appbelow}_j, & j\le i,\\
        {\appabove}_i & j=n+1.
    \end{cases}
\end{equation}


By Lemma~\ref{lemma:distance}, and Corollary~\ref{cor:threshold-lemma-bound-on-v-b}, the choice of values for the matrix $A$ in Equation~\ref{def:set-coeff-matrix} indeed provides the required lower bounds on the approximation ratio of each buyer-offering mechanism on each interval $I_j$.


\subsubsection*{Step 3: Variables and the Linear Program}
Consider an arbitrary distribution $\distbuyer$ over the buyer's value. 
We define weights $\vec{w} = (w_1, \dots, w_{n+1})$ over the intervals $I_1, \dots, I_{n+1}$ such that each weight represents the contribution of the values in the corresponding interval to the optimal welfare. That is, for each interval $I_j$, let:
\begin{equation*}
    w_j  = \frac{\int_{v_b \in I_j} \distbuyerpdf(v_b)\cdot \opt{\distribution, \distbuyer}(v_b) \, dv_b}{\opt{\distribution, \distbuyer}}
\end{equation*}


Then, for every reserve $r_i$, the quantity $(A \cdot \vec{w}^\top)_i$ provides a lower bound on the approximation ratio of the buyer-offering mechanism with reserve $r_i$ under the distribution $\distbuyer$. To see this, suppose we have a lower bound $\alpha_j$ on the approximation ratio of a mechanism $\mechanism$ whenever the buyer's value falls in interval $I_j$. Then the overall approximation ratio of $\mechanism$ under the buyer distribution $\distbuyer$ satisfies
\begin{align*}
    \frac{\mechanism_{\distribution, \distbuyer}}{\opt{\distribution, \distbuyer}}
    & = \frac{\int_{v_b \in \suppbuyer} \mechanism_{\distribution, \distbuyer}(v_b)\cdot \distbuyerpdf(v_b) \, dv_b}{\opt{\distribution, \distbuyer}} \\
    & = \sum_{j \in [n+1]} \frac{\int_{v_b \in I_j} \mechanism_{\distribution, \distbuyer}(v_b)\cdot \distbuyerpdf(v_b) \, dv_b}{\opt{\distribution, \distbuyer}} \\
    & \ge \sum_{j \in [n+1]} \alpha_j \cdot w_j,
\end{align*}
where the inequality follows from the interval-wise lower bounds $\alpha_j$.

Observe that the weights satisfy $w_j \in [0,1]$ and $\sum_{j \in [n+1]} w_j = 1$. Therefore, by treating the $w_j$ as variables over the simplex, we can compute the worst-case choice of weights. The resulting linear program provides a lower bound on the approximation ratio of the buyer-offering mechanism for any buyer and seller distributions.

\begin{equation}\label{program:lp-primal}
\begin{aligned}
    & \text{minimize } && t \\
    & \text{subject to } && A \cdot \Vec{w}^\top \le t \cdot \Vec{1}_n, \\
    & && w_j \ge 0 \quad \text{for all } j \in [n+1], \\
    & && \sum_{j \in [n+1]} w_j = 1.
\end{aligned}
\end{equation}

\begin{remark}\label{remark:exact-seq}
    Note that when ${\appbelow}_{\min} > 0$, then taking $\reserve = 1$ ensures that the approximation ratio of the buyer-offering mechanism with reserve $\reserve$ is at least ${\appbelow}_{\min}$ when $v_b \leq \reserve$, and is equal to $1$ when $v_b > \reserve$ (since the buyer is forced to choose an offer at least as large as the seller's value). 
   In this case, if ${\appbelow}_{\min} > {\appbelow}_n$, then for all indices $j$ with ${\appbelow}_j < {\appbelow}_{\min}$ we may replace the pair $({\appbelow}_j, {\appabove}_j)$ by $({\appbelow}_{\min}, 1)$. Since
\[
({\appbelow}_j, {\appabove}_j) \leq ({\appbelow}_{\min}, 1)
\quad\text{for all such } j,
\]
this replacement results in a coefficient matrix $A'$ that dominates $A$ entry-wise, i.e., $A'_{i,j} \geq A_{i,j}$ for all $i,j$. Note that this construction may result in repeated intervals, or intervals with a weight of $0$; however, this can only improve the resulting approximation guarantee.
\end{remark}

This completes the proof of Proposition~\ref{prop:lp-bound-given-sigma}.

\subsection*{Computing a Lower Bound on the Value of the Linear Program}

\subsubsection*{The Dual Problem}

To prove a lower bound on our linear program, we use weak duality. For this purpose, we define the dual problem~\ref{program:lp-dual}, and refer to the linear program in~\ref{program:lp-primal} as the primal problem.

\begin{equation}\label{program:lp-dual}
\begin{aligned}
    & \text{maximize } \lambda \\
    & \text{subject to } & \quad  A^t\Vec{y} \geq  \lambda \cdot \Vec{1} & \\
    & & y_i \geq 0 \quad & \text{ for every } i \in [n]\\
    & & \sum_{i \in [n]} y_i =1
\end{aligned}
\end{equation}

Let $t, \Vec{w}$ be a feasible solution to the primal linear program, and let $\lambda, \Vec{y}$ be a feasible solution to its dual. By weak duality, we have $ \lambda \leq t$. Therefore, to prove that the optimal value of the primal problem is at least some value $b$, it suffices to construct a feasible dual solution $\lambda, \Vec{y}$ such that $\lambda \geq b$.  

The exact values used to obtain the approximation ratio of $\upperbound$ are presented in Appendix~\ref{sec:exact-lp-values}.

%% file: lp-instance.tex
\section{Constructing the Coefficient Matrix}\label{sec:exact-lp-values}

In this appendix, we give the numerical values used to prove that the value of
the linear program in Section~5 is at least \(0.746\). We stress that the displayed numerical values are the values used in the construction of the
linear program and in the dual solution below, rather than shorthand for hidden
higher-precision values. Each numerical entry of the coefficient matrix is a
valid lower bound on the corresponding approximation guarantee, and the
feasibility of the dual solution is checked using the displayed values.

We now construct the specific coefficient matrix that we use to prove our main theorem.  Set $n=19$, and let 
\begin{equation}\label{def:alpha-left-vals}
    {\vec{\appbelow}}^T = \begin{pmatrix}
0.782 \\
0.746 \\
0.728 \\
0.710 \\
0.692 \\
0.674 \\
0.657 \\
0.640 \\
0.622 \\
0.605 \\
0.589 \\
0.572 \\
0.539 \\
0.507 \\
0.476 \\
0.446 \\
0.416 \\
0.345 \\
0.280 \\
\end{pmatrix}
\end{equation}

As described in Section~\ref{sec:coeff-matrix} the values ${\appabove}_1, \dots, {\appabove}_n$ are computed as the solutions to optimization problem in Corollary~\ref{cor:threshold-lemma-no-bound-on-v-b}. We explicitly show the computation of the one value, the remaining values follow similarly. 

The approximation ratio of the buyer-offering mechanism with reserve $\reserve $ , chosen so that the approximation ratio of the no-trade mechanism at buyer value $r$ is at most ${\appbelow}_1 = 0.746$, is at least
\begin{equation*}
\phi(0.746)\;=\;
\min_{0.254 \leq x\leq 1}
\left(
1+x\ln x
-
\frac{\ln x \left(1-x+x\ln x\right)}{\ln x +\frac{0.746}{0.254}}
\right).
\end{equation*}

We now show that $\phi(0.746) \ge 0.7$. Similar to Section~\ref{sec-seller-only}, let
\[
h(x)\;=\;
1+x\ln x
-
\frac{\ln x \left(1-x+x\ln x\right)}{\ln x +\frac{0.746}{0.254}}
\]
On the interval $x\in[0.254,1]$, the function $h(x)$ is smooth and continuously differentiable. Its derivative is given by Equation~\ref{equ:deriviative-h} with $\frac{\alpha}{1-\alpha} = \frac{0.746}{0.254}$.

A direct numerical inspection shows that $h'(x)$ has a unique zero in the interval
$x\in[0.254,1]$, occurring at $x^\star \approx 0.518$.
Moreover, $h'(x)<0$ for $x<x^\star$ and $h'(x)>0$ for $x>x^\star$, implying that
$x^\star$ is the unique minimizer of $h(x)$ on this interval.
Evaluating the function at this point yields $h(x^\star)\approx 0.70004$.
Thus, ${\appabove}_1 \geq 0.7$. Similarly, we obtain the vector $\vec{\appabove}$, whose exact values are in Figure~\ref{fig:vec-app-above}:

\begin{figure}[!htbp]
    \centering
    \[
{\vec{\appabove}}^T = \begin{pmatrix}
0.690 \\
0.700 \\
0.705 \\
0.710 \\
0.715 \\
0.720 \\
0.725 \\
0.730 \\
0.735 \\
0.740 \\
0.745 \\
0.750 \\
0.760 \\
0.770 \\
0.780 \\
0.790 \\
0.800 \\
0.825 \\
0.850 \\
\end{pmatrix}
\]    \caption{The vector $\vec{\appabove}$}
    \label{fig:vec-app-above}
\end{figure}

Then, we compute the values $\lambda_{i,j}, \alpha_{i,j}$ for $i< j \leq n$ and the coefficient matrix $A$ as is described in Section~\ref{sec:coeff-matrix}, see figure~\ref{fig:matrices} for exact values.  


\begin{figure}[p]
    \centering
\noindent\resizebox{\textwidth}{!}{$\lambda = \begin{pmatrix}
- & 0.818 & 0.746 & 0.682 & 0.626 & 0.576 & 0.533 & 0.495 & 0.458 & 0.426 & 0.399 & 0.372 & 0.325 & 0.286 & 0.253 & 0.224 & 0.198 & 0.146 & 0.108 \\
- & - & 0.911 & 0.833 & 0.764 & 0.703 & 0.652 & 0.605 & 0.560 & 0.521 & 0.487 & 0.455 & 0.398 & 0.350 & 0.309 & 0.274 & 0.242 & 0.179 & 0.132 \\
- & - & - & 0.914 & 0.839 & 0.772 & 0.715 & 0.664 & 0.614 & 0.572 & 0.535 & 0.499 & 0.436 & 0.384 & 0.339 & 0.300 & 0.266 & 0.196 & 0.145 \\
- & - & - & - & 0.917 & 0.844 & 0.782 & 0.726 & 0.672 & 0.625 & 0.585 & 0.545 & 0.477 & 0.420 & 0.371 & 0.328 & 0.290 & 0.215 & 0.158 \\
- & - & - & - & - & 0.920 & 0.852 & 0.791 & 0.732 & 0.681 & 0.637 & 0.594 & 0.520 & 0.457 & 0.404 & 0.358 & 0.317 & 0.234 & 0.173 \\
- & - & - & - & - & - & 0.926 & 0.859 & 0.795 & 0.740 & 0.693 & 0.646 & 0.565 & 0.497 & 0.439 & 0.389 & 0.344 & 0.254 & 0.188 \\
- & - & - & - & - & - & - & 0.928 & 0.859 & 0.799 & 0.748 & 0.697 & 0.610 & 0.536 & 0.474 & 0.420 & 0.371 & 0.274 & 0.203 \\
- & - & - & - & - & - & - & - & 0.925 & 0.861 & 0.806 & 0.751 & 0.657 & 0.578 & 0.510 & 0.452 & 0.400 & 0.296 & 0.218 \\
- & - & - & - & - & - & - & - & - & 0.930 & 0.870 & 0.812 & 0.710 & 0.624 & 0.552 & 0.489 & 0.432 & 0.320 & 0.236 \\
- & - & - & - & - & - & - & - & - & - & 0.935 & 0.872 & 0.763 & 0.671 & 0.593 & 0.525 & 0.465 & 0.343 & 0.253 \\
- & - & - & - & - & - & - & - & - & - & - & 0.932 & 0.815 & 0.717 & 0.633 & 0.561 & 0.497 & 0.367 & 0.271 \\
- & - & - & - & - & - & - & - & - & - & - & - & 0.874 & 0.769 & 0.679 & 0.602 & 0.533 & 0.394 & 0.290 \\
- & - & - & - & - & - & - & - & - & - & - & - & - & 0.879 & 0.776 & 0.688 & 0.609 & 0.450 & 0.332 \\
- & - & - & - & - & - & - & - & - & - & - & - & - & - & 0.883 & 0.782 & 0.692 & 0.512 & 0.378 \\
- & - & - & - & - & - & - & - & - & - & - & - & - & - & - & 0.886 & 0.784 & 0.579 & 0.428 \\
- & - & - & - & - & - & - & - & - & - & - & - & - & - & - & - & 0.884 & 0.654 & 0.483 \\
- & - & - & - & - & - & - & - & - & - & - & - & - & - & - & - & - & 0.739 & 0.545 \\
- & - & - & - & - & - & - & - & - & - & - & - & - & - & - & - & - & - & 0.738 \\
- & - & - & - & - & - & - & - & - & - & - & - & - & - & - & - & - & - & - \\
\end{pmatrix}$}
\vspace{1em}

\resizebox{\textwidth}{!}{$\alpha =
\begin{pmatrix}
- & 0.933 & 0.906 & 0.883 & 0.862 & 0.844 & 0.828 & 0.814 & 0.800 & 0.789 & 0.779 & 0.769 & 0.752 & 0.737 & 0.725 & 0.714 & 0.705 & 0.686 & 0.672 \\
- & - & 0.967 & 0.938 & 0.913 & 0.891 & 0.872 & 0.854 & 0.838 & 0.823 & 0.811 & 0.799 & 0.778 & 0.760 & 0.745 & 0.732 & 0.721 & 0.698 & 0.680 \\
- & - & - & 0.968 & 0.940 & 0.916 & 0.895 & 0.876 & 0.858 & 0.842 & 0.829 & 0.815 & 0.792 & 0.773 & 0.756 & 0.742 & 0.730 & 0.704 & 0.685 \\
- & - & - & - & 0.969 & 0.942 & 0.919 & 0.899 & 0.879 & 0.862 & 0.847 & 0.832 & 0.807 & 0.786 & 0.768 & 0.753 & 0.739 & 0.711 & 0.690 \\
- & - & - & - & - & 0.970 & 0.945 & 0.923 & 0.901 & 0.882 & 0.866 & 0.850 & 0.823 & 0.800 & 0.780 & 0.763 & 0.748 & 0.718 & 0.695 \\
- & - & - & - & - & - & 0.972 & 0.948 & 0.924 & 0.904 & 0.887 & 0.869 & 0.840 & 0.815 & 0.793 & 0.775 & 0.758 & 0.725 & 0.701 \\
- & - & - & - & - & - & - & 0.973 & 0.948 & 0.926 & 0.907 & 0.888 & 0.856 & 0.829 & 0.806 & 0.786 & 0.768 & 0.733 & 0.706 \\
- & - & - & - & - & - & - & - & 0.972 & 0.949 & 0.928 & 0.908 & 0.874 & 0.844 & 0.820 & 0.798 & 0.779 & 0.741 & 0.712 \\
- & - & - & - & - & - & - & - & - & 0.974 & 0.952 & 0.930 & 0.893 & 0.862 & 0.835 & 0.812 & 0.791 & 0.749 & 0.719 \\
- & - & - & - & - & - & - & - & - & - & 0.976 & 0.953 & 0.913 & 0.879 & 0.850 & 0.825 & 0.803 & 0.759 & 0.726 \\
- & - & - & - & - & - & - & - & - & - & - & 0.975 & 0.932 & 0.896 & 0.866 & 0.839 & 0.816 & 0.768 & 0.733 \\
- & - & - & - & - & - & - & - & - & - & - & - & 0.954 & 0.916 & 0.883 & 0.855 & 0.830 & 0.779 & 0.742 \\
- & - & - & - & - & - & - & - & - & - & - & - & - & 0.957 & 0.920 & 0.888 & 0.860 & 0.803 & 0.761 \\
- & - & - & - & - & - & - & - & - & - & - & - & - & - & 0.959 & 0.924 & 0.892 & 0.829 & 0.783 \\
- & - & - & - & - & - & - & - & - & - & - & - & - & - & - & 0.961 & 0.926 & 0.857 & 0.806 \\
- & - & - & - & - & - & - & - & - & - & - & - & - & - & - & - & 0.962 & 0.886 & 0.830 \\
- & - & - & - & - & - & - & - & - & - & - & - & - & - & - & - & - & 0.918 & 0.857 \\
- & - & - & - & - & - & - & - & - & - & - & - & - & - & - & - & - & - & 0.927 \\
- & - & - & - & - & - & - & - & - & - & - & - & - & - & - & - & - & - & - \\
\end{pmatrix}$}
\vspace{1em}

\resizebox{\textwidth}{!}{$ A =
\begin{pmatrix}
0.782 & 0.933 & 0.906 & 0.883 & 0.862 & 0.844 & 0.828 & 0.814 & 0.800 & 0.789 & 0.779 & 0.769 & 0.752 & 0.737 & 0.725 & 0.714 & 0.705 & 0.690 & 0.690 & 0.690 \\
0.782 & 0.746 & 0.967 & 0.938 & 0.913 & 0.891 & 0.872 & 0.854 & 0.838 & 0.823 & 0.811 & 0.799 & 0.778 & 0.760 & 0.745 & 0.732 & 0.721 & 0.700 & 0.700 & 0.700 \\
0.782 & 0.746 & 0.728 & 0.968 & 0.940 & 0.916 & 0.895 & 0.876 & 0.858 & 0.842 & 0.829 & 0.815 & 0.792 & 0.773 & 0.756 & 0.742 & 0.730 & 0.705 & 0.705 & 0.705 \\
0.782 & 0.746 & 0.728 & 0.710 & 0.969 & 0.942 & 0.919 & 0.899 & 0.879 & 0.862 & 0.847 & 0.832 & 0.807 & 0.786 & 0.768 & 0.753 & 0.739 & 0.711 & 0.710 & 0.710 \\
0.782 & 0.746 & 0.728 & 0.710 & 0.692 & 0.970 & 0.945 & 0.923 & 0.901 & 0.882 & 0.866 & 0.850 & 0.823 & 0.800 & 0.780 & 0.763 & 0.748 & 0.718 & 0.715 & 0.715 \\
0.782 & 0.746 & 0.728 & 0.710 & 0.692 & 0.674 & 0.972 & 0.948 & 0.924 & 0.904 & 0.887 & 0.869 & 0.840 & 0.815 & 0.793 & 0.775 & 0.758 & 0.725 & 0.720 & 0.720 \\
0.782 & 0.746 & 0.728 & 0.710 & 0.692 & 0.674 & 0.657 & 0.973 & 0.948 & 0.926 & 0.907 & 0.888 & 0.856 & 0.829 & 0.806 & 0.786 & 0.768 & 0.733 & 0.725 & 0.725 \\
0.782 & 0.746 & 0.728 & 0.710 & 0.692 & 0.674 & 0.657 & 0.640 & 0.972 & 0.949 & 0.928 & 0.908 & 0.874 & 0.844 & 0.820 & 0.798 & 0.779 & 0.741 & 0.730 & 0.730 \\
0.782 & 0.746 & 0.728 & 0.710 & 0.692 & 0.674 & 0.657 & 0.640 & 0.622 & 0.974 & 0.952 & 0.930 & 0.893 & 0.862 & 0.835 & 0.812 & 0.791 & 0.749 & 0.735 & 0.735 \\
0.782 & 0.746 & 0.728 & 0.710 & 0.692 & 0.674 & 0.657 & 0.640 & 0.622 & 0.605 & 0.976 & 0.953 & 0.913 & 0.879 & 0.850 & 0.825 & 0.803 & 0.759 & 0.740 & 0.740 \\
0.782 & 0.746 & 0.728 & 0.710 & 0.692 & 0.674 & 0.657 & 0.640 & 0.622 & 0.605 & 0.589 & 0.975 & 0.932 & 0.896 & 0.866 & 0.839 & 0.816 & 0.768 & 0.745 & 0.745 \\
0.782 & 0.746 & 0.728 & 0.710 & 0.692 & 0.674 & 0.657 & 0.640 & 0.622 & 0.605 & 0.589 & 0.572 & 0.954 & 0.916 & 0.883 & 0.855 & 0.830 & 0.779 & 0.750 & 0.750 \\
0.782 & 0.746 & 0.728 & 0.710 & 0.692 & 0.674 & 0.657 & 0.640 & 0.622 & 0.605 & 0.589 & 0.572 & 0.539 & 0.957 & 0.920 & 0.888 & 0.860 & 0.803 & 0.761 & 0.760 \\
0.782 & 0.746 & 0.728 & 0.710 & 0.692 & 0.674 & 0.657 & 0.640 & 0.622 & 0.605 & 0.589 & 0.572 & 0.539 & 0.507 & 0.959 & 0.924 & 0.892 & 0.829 & 0.783 & 0.770 \\
0.782 & 0.746 & 0.728 & 0.710 & 0.692 & 0.674 & 0.657 & 0.640 & 0.622 & 0.605 & 0.589 & 0.572 & 0.539 & 0.507 & 0.476 & 0.961 & 0.926 & 0.857 & 0.806 & 0.780 \\
0.782 & 0.746 & 0.728 & 0.710 & 0.692 & 0.674 & 0.657 & 0.640 & 0.622 & 0.605 & 0.589 & 0.572 & 0.539 & 0.507 & 0.476 & 0.446 & 0.962 & 0.886 & 0.830 & 0.790 \\
0.782 & 0.746 & 0.728 & 0.710 & 0.692 & 0.674 & 0.657 & 0.640 & 0.622 & 0.605 & 0.589 & 0.572 & 0.539 & 0.507 & 0.476 & 0.446 & 0.416 & 0.918 & 0.857 & 0.800 \\
0.782 & 0.746 & 0.728 & 0.710 & 0.692 & 0.674 & 0.657 & 0.640 & 0.622 & 0.605 & 0.589 & 0.572 & 0.539 & 0.507 & 0.476 & 0.446 & 0.416 & 0.345 & 0.927 & 0.825 \\
0.782 & 0.746 & 0.728 & 0.710 & 0.692 & 0.674 & 0.657 & 0.640 & 0.622 & 0.605 & 0.589 & 0.572 & 0.539 & 0.507 & 0.476 & 0.446 & 0.416 & 0.345 & 0.280 & 0.850 \\
\end{pmatrix}$}
    \caption{Numeric values of $\lambda_{i,j}$, $\alpha_{i,j}$, and the coefficient matrix $A$ corresponding to the instance with 
$\vec{\appbelow}$ as defined in Equation~\eqref{def:alpha-left-vals}. All values are truncated to three decimal places.}
    \label{fig:matrices}
\end{figure}

\begin{claim}\label{claim:sol-to-lp}
    The value of the linear program ($\ref{program:lp-primal}$) when defined with the matrix $A^*$ is at least $\upperbound$.
\end{claim}

\begin{proof}[Proof of Claim~\ref{claim:sol-to-lp}]
    To prove this claim, we construct a feasible solution to the dual linear program. 
    Let $\lambda = 0.746$, and:
    \[ \vec{y} = 
    \begin{pmatrix}
    0.00003 \\
    0.07532 \\
    0.07298 \\
    0.06952 \\
    0.06542 \\
    0.05836 \\
    0.05455 \\
    0.05423 \\
    0.04791 \\
    0.04184 \\
    0.04223 \\
    0.07492 \\
    0.06335 \\
    0.05387 \\
    0.04573 \\
    0.04016 \\
    0.08034 \\
    0.03813 \\
    0.02111
    \end{pmatrix}
    \]
    One can verify that the entries of $y$ sum to $1$ and that $A^t \cdot \vec{y} \geq \lambda\cdot\vec{1}$, confirming feasibility.

\end{proof}